\documentclass[pdflatex,sn-mathphys-num]{sn-jnl}

\usepackage{lmodern}
\usepackage{graphicx}%
\usepackage{multirow}%
\usepackage{amsmath,amssymb,amsfonts}%
\usepackage{amsthm}%
\usepackage{mathrsfs}%

\usepackage[title]{appendix}%
\usepackage{xcolor}%
\usepackage{textcomp}%
\usepackage{manyfoot}%
\usepackage{booktabs}%
\usepackage{algorithm}%
\usepackage{algorithmicx}%
\usepackage{algpseudocode}%
\usepackage{listings}%
\usepackage{graphics}
\usepackage{subfigure}
\usepackage{epstopdf}
\usepackage{csquotes}
\newcommand{\ket}[1]{|#1\rangle}
\newcommand{\bra}[1]{\langle#1|}
\DeclareMathOperator{\trace}{tr}
\usepackage{mathrsfs}
\usepackage{mathtools}
\usepackage{amsmath}

\newtheorem{theorem}{Theorem}
\begin{document}

\title[Article Title]{Learnability of a hybrid quantum-classical neural network for graph-structured quantum data}

\author[1]{\fnm{Yanying} \sur{Liang}}

\author[1]{\fnm{Sile} \sur{Tang}}

\author[1]{\fnm{Zhehao} \sur{Yi}}

\author[1]{\fnm{Haozhen} \sur{Situ}}

\author*[2,3]{\fnm{Zhu-Jun} \sur{Zheng}}\email{zhengzj@scut.edu.cn}

\affil*[1]{\orgdiv{College of Mathematics and Informatics}, \orgname{South China Agricultural University}, \city{Guangzhou}, \postcode{510642}, \country{China}}

\affil[2]{\orgdiv{College of Mathematics}, \orgname{South China University of Technology},  \city{Guangzhou}, \postcode{510642}, \country{China}}

\affil[3]{\orgdiv{Laboratory of Quantum Science and Engineering}, \orgname{South China University of Technology}, \city{Guangzhou}, \postcode{510642}, \country{China}}

\abstract{
Graph-structured data commonly arise in many real-world applications, and this extends naturally into the quantum setting, where quantum data with inherent graph structures are frequently generated by typical quantum data sources. However, existing state-of-the-art models often lack training and evaluation on deeper quantum neural networks. In this work, we design a hybrid quantum-classical neural network with deep residual learning, termed Res-HQCNN, specifically designed to handle graph-structured quantum data.
Building upon this architecture, we systematically explore the interplay between residual block structures and graph information in both training and testing phases. Through extensive experiments, we demonstrate that incorporating graph structure information into the quantum data significantly improves learning efficiency compared to the existing model. Additionally, we conduct comparative experiments to evaluate the effectiveness of residual blocks. Our results show that the residual structure enables deeper Res-HQCNN models to learn graph-structured quantum data more efficiently and accurately.
}


\maketitle

\section{Introduction}\label{sec1}

Graph structures are prevalent in natural science, which are widely used when representing many-to-many relationships. Mathematically, a graph is a data structure denoted as $G=(V,E)$, where $V$ is a set of vertices and $E$ is a set of edges. In the context of quantum data, graph-structured quantum data refers to a specialized form of quantum information generated by structured quantum devices. The graph structure can be imposed on quantum data by associating a classical graph $G=(V,E)$ with the density operators of quantum states in a Hilbert space. For example, if we consider a distributed set of quantum information processors and a graph $G=(V,E)$, then there exists a quantum information processor at some vertex $i$ of vertex $V$ , which takes a quantum state $\rho_i$ as an input. A symbol $\sigma_i$ is an ideal output corresponding to $\rho_i$ in a training set $\{(\rho_i,\sigma_i)|i=1,2,\cdots, N\}$. As for the edges in $E$, they capture the connectivity structure of training data, which quantifies the correlations of two neighboring quantum states.
 
Quantum machine learning (QML) has the potential to improve the analysis of quantum data or classical data due to the use of quantum information theory \cite{biamonte2017quantum,cerezo2022challenges}.
Up to now, there already exists some researches for graph-structured data using QML method \cite{cong2019,gui2019,ai2022decompositional,coushu_albrecht2023quantum,dernbach2019quantum,hu2022design,mernyei2022equivariant,shah2021quantum,zheng2021quantum,choi2021tutorial,zhang2019quantum}. After comparing the different ideas in the aforementioned articles, we observe that they mainly focus on integrating the learning capabilities of quantum neural networks (QNNs) with the topological information of graphs within quantum circuits. In this way, the flexible models proposed in these studies not only preserve the essential relationships within the data but also help reduce information loss. However, during experimental implementation, these models often require preprocessing steps that convert classical data into quantum data, typically using amplitude encoding or qubit encoding, which increases computational complexity. To overcome this limitation, we aim to directly utilize graph-structured quantum data to explore the potential quantum advantages of QNNs. 

Notably, in 2023, Beer \emph{et al.} attempted to use universal QNNs to learn graph-structured quantum data \cite{beer2023quantum}. 
By designing an information-theoretic cost function to capture the graph structure of quantum data, they prove that the graph structure in quantum data can lead to better learning of QNNs.
In general, to achieve better training performance, the depth of the neural network is often increased, as it is well understood that a deeper neural network cannot perform worse than a shallower one in terms of representation capacity \cite{he2016deep}. However, increasing the depth of the neural network may lead to a deterioration in cost function performance, due to degradation problems such as vanishing gradients or barren plateaus. To solve this kind of problems in QNNs, a hybrid quantum-classical neural network with deep residual learning (Res-HQCNN) was proposed to improve the numerical results for deeper QNNs \cite{liang2021hybrid}. Making use of the power of residual block structure \cite{he2016deep}, Res-HQCNN has better ability to learn an unitary transformation, and has stronger robustness for noise quantum data. As far as we know, Res-HQCNN is the first model to transform deep residual block structure into quantum concept and showed its power in deeper QNNs.
Later, other state-of-the-art models emerged that adopted the residual approach \cite{kashif2024resqnets,wen2024enhancing}.

\begin{figure*}[htbp]
\centering
\includegraphics[width=0.9\textwidth]{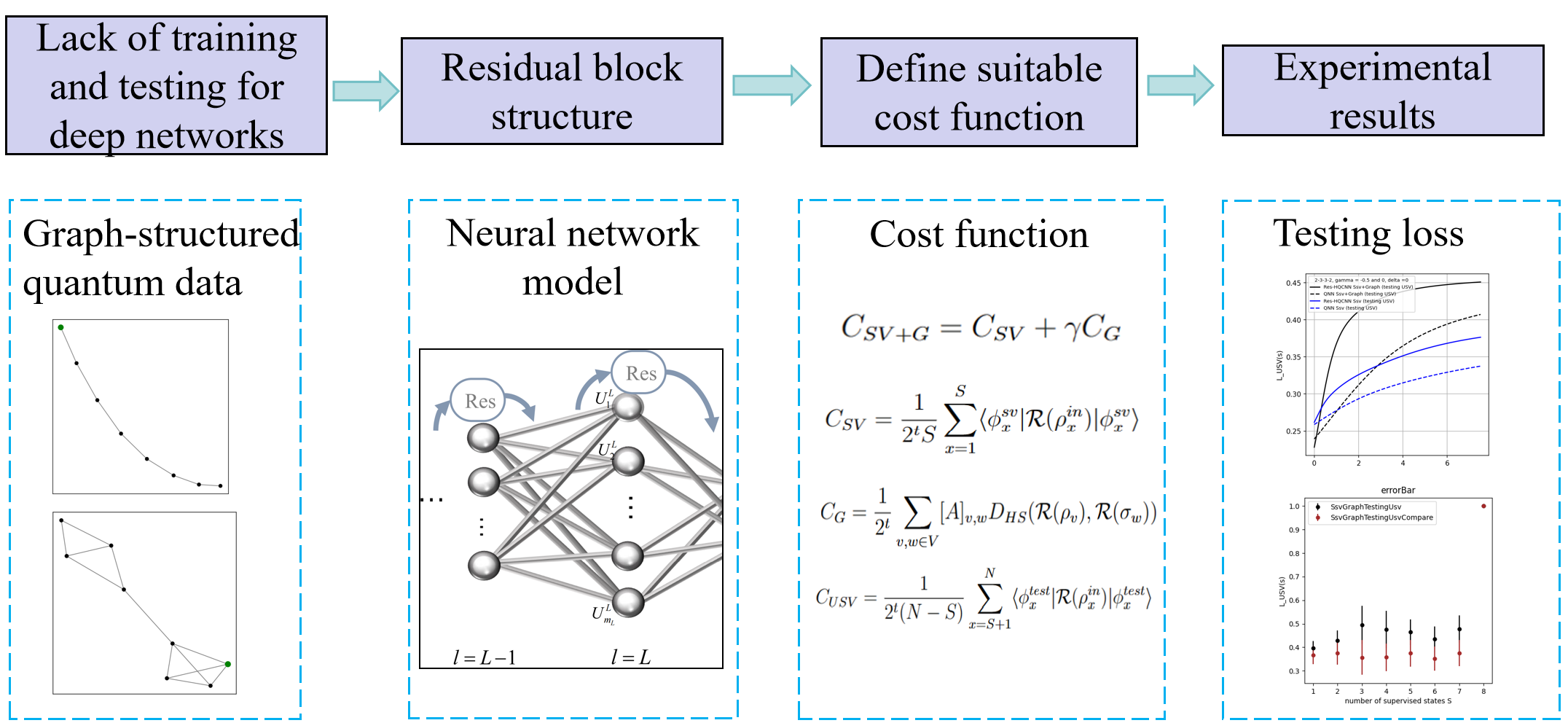}
\caption{Overall block diagram of this paper}\label{figure0}
\end{figure*}

In this paper, we will explore the potential of residual approaches for handling graph-structured quantum data in deeper QNNs.
We aim to evaluate whether Res-HQCNN can effectively process graph-structured quantum data and deliver promising experimental results. Specifically, the first objective of this paper is to demonstrate that incorporating graph structure information in quantum data can enhance the learning efficiency of QNNs, compared to existing state-of-the-art models. The second objective is to show that, when applied to deeper QNN architectures, residual block structures can further improve the network's ability to learn from graph-structured quantum data.
This exploration presents both challenges and opportunities. The overall framework of the paper is illustrated in Fig.\ref{figure0}, and the main contributions are as follows:
 \begin{itemize}
     \item Propose the model of Res-HQCNN with graph-structured quantum data, and define the cost function for quantum data with or without graph. 
     \item Present the training algorithm with analysis from the perspective of propagating information feedforward and backward.
     \item A comprehensive set of experiments is designed for Res-HQCNN, with and without graph-structured quantum data, to rigorously assess the contribution of graph information to learning effectiveness.
     \item Design experiments to evaluate the performance of Res-HQCNN with graph-structured quantum data at varying network depths, aiming to demonstrate the effectiveness of the residual block structure.
 \end{itemize}

 The remainder of this paper is organized as follows. Section \ref{22} reviews related works on QNNs dealing with graph-structured classical and quantum data, as well as deep residual networks from both classical and quantum perspectives.  Section \ref{44} presents our approach to learning graph-structured quantum data using Res-HQCNN, including detailed descriptions of the model architecture and the training algorithm for semi-supervised quantum states with or without graph. To test the learning efficiency, Section \ref{5} reports experimental simulations along with comprehensive analysis. Finally, section \ref{6} concludes the paper and offers discussions on future directions.

\section{Related works}\label{22}

\subsection{QNNs for graph-structured classical and quantum data}

In the case of graph-structured classical data, it is essential to develop suitable encoding strategies that map the data into quantum states, enabling QNNs to process it with minimal information loss.
In 2019, a quantum walk neural network was proposed, which learns coin operators determining the quantum random walks \cite{dernbach2019quantum}. In 2021, the authors in \cite{chen2021hybrid} presented a hybrid quantum-classical graph convolutional neural network to learn high energy physics data, deepening the research of high-energy physics.
Hu \emph{et al.} chose the Givens rotations and corresponding quantum implementation to encode graph information \cite{hu2022design}. And based on variational quantum circuits, they designed an efficient quantum graph convolutional neural network to conduct semi-supervised learning on graph-structured classical data. In 2023, Andrea Skolik \emph{et al.} made use of an equivariant quantum circuit for learning tasks on weighted graphs, emphasizing the important role of symmetry-preserving ansatzes in QML \cite{skolik2023equivariant}. Recently, a quantum graph neural network has been proposed to predict the properties of chemistry and physics for materials \cite{ryu2023quantum}.

When it comes to graph-structured quantum data, in 2023, Beer \emph{et al.} described the definition of graph-structured quantum data formally, and showed how to use a pure quantum neural network to learn graph-structured quantum data \cite{beer2023quantum}. 
But they only tested the learning ability of QNNs with two or three layers. As we know, deeper QNNs could perform better than the shallower ones. Moreover, in order to avoid the degradation problems in deep QNNs, the residual approach should be considered. So combining all of the above considerations, we decide to use Res-HQCNN proposed by us in Ref.\cite{liang2021hybrid} to learn graph-structured quantum data, especially in deep architecture.

\subsection{Residual block structure in classical and quantum neural networks}

\begin{figure*}[htbp]
\centering
\subfigure[In classical neural network]{
\includegraphics[width=9cm]{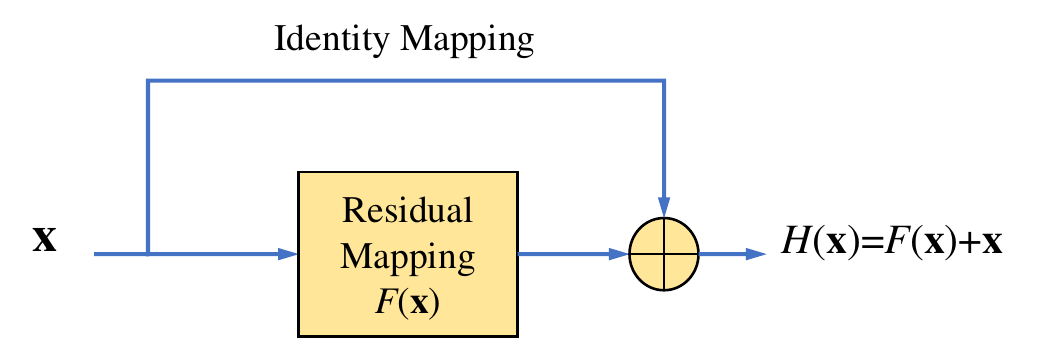}
}
\subfigure[In quantum neural network]{
\includegraphics[width=10cm]{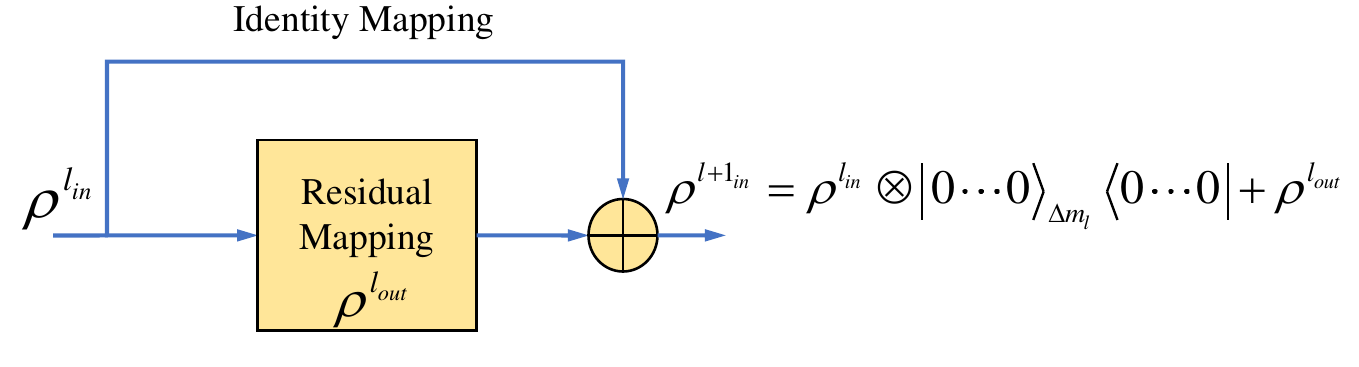}
}
\caption{ \textbf{Residual block structure in classical and quantum neural network}}
\label{Figure1}
\end{figure*} 

In classical neural networks, deeper neural networks are not easy to train due to possible degradation problems. In 2016, residual block structure was proposed to solve the degradation problems \cite{he2016deep}, which is shown in Fig.\ref{Figure1}(a). The classical residual block structure can deal with hundreds or thousands of layers in one neural network with satisfying results of training and testing loss. Later, researchers proposed many variants of deep residual networks, and also obtained the state-of-the-art experimental performance \cite{jian2016deep,lee2018deep,wu2018deep,zhang2019deep,shafiq2022deep,zhang2023selective,alenezi2023wavelet}. 

In QNNs, a novel definition of residual block structure was proposed in quantum concept \cite{liang2021hybrid}, which is shown in Fig.\ref{Figure1}(b). In  Fig.\ref{Figure1}(b), $m_l$ represents the number of nodes in layer $l$, $\rho^{l_{in}}$ and $\rho^{l_{out}}$ mean the input and output quantum state for layer $l$. The corresponding experimental simulations in \cite{liang2021hybrid} demonstrate the power of residual block structure for both clean and noisy quantum data in deep QNNs. 

In this paper, when the architecture of Res-HQCNN has more layers, we expect the residual block structure in Fig.\ref{Figure1}(b) will show its power for graph-structured quantum data.

\section{Res-HQCNN learns graph-structured quantum data}\label{44}

\subsection{Graph-structured quantum data}

Assume the quantum system in this paper is isomorphic to $(\mathbb{C}^2)^{\otimes m}$. The quantum
device produce the quantum states with structure in some probability distribution, which leads to $(\rho_{v_j},p_{v_j})$ with $j=1,2,\cdots,n$ for quantum sources over set $\{\rho_{v_1},\rho_{v_2},\cdots,\rho_{v_n}\}$.
Specifically, a graph-structured quantum data can be explained as follows. Consider a graph $G=(V,E)$ with nodes $V$ and edges $E$, there is a map $\rho: V \xrightarrow{} \mathcal{D(H)}$ 
transforming a vertex $v$
to a density operator $\rho_v$ on $ \mathcal{H}$.
The edge set $E$
describes the connectivity structure of the quantum data, and quantifies the information-theoretic closeness between neighboring quantum states. Here 
we define that two states $\rho_v$ and $\rho_w$ are neighboring for the edge between vertex $v$ and $w$ if they are close in a given information metric $d$, that is $d(\rho_v,\rho_w)\leq \epsilon$ \cite{beer2023quantum}. Assume that we know which graph-structured quantum data are close, so that we can quantify and explore the information-theoretic closeness of these quantum data through experimental setup.

Without loss of generality, the training data includes $S$ supervised vertices firstly and $N-S$ unsupervised vertices with $0\leq S \leq N$. So the full quantum data set for training and testing tasks can be written as
\begin{align}\label{data}
\{(\rho_1,\ket{\phi_1^{sv}}\bra{\phi_1^{sv}}),\cdots,(\rho_S,\ket{\phi_S^{sv}}\bra{\phi_S^{sv}}), \nonumber\\
(\rho_{S+1},\ket{\phi_{S+1}^{test}}\bra{\phi_{S+1}^{test}}),\cdots,(\rho_{N},\ket{\phi_{N}^{test}}\bra{\phi_{N}^{test}})\}
\end{align}

The supervised training data are possibly unknown generated quantum states in the form of $(\rho_S,V\ket{\phi_S^{sv}}\bra{\phi_S^{sv}}))$ for $x=1,2,\cdots,S$ with an unknown unitary matrix $V$. The elements of $\ket{\phi_S^{sv}}$ are randomly selected from a normal distribution before normalization. The elements of unitary matrix $V$ are randomly selected from a normal distribution before orthogonalization. 

\subsection{Res-HQCNN with graph-structured quantum data}

Res-HQCNN is a powerful hybrid quantum-classical neural network that combines the expressivity of quantum neural networks with the training stability introduced by residual connections. This architecture has been shown to efficiently approximate arbitrary unitary transformations and demonstrates strong robustness against noise in quantum data \cite{liang2021hybrid}. In this work, we apply Res-HQCNN to graph-structured quantum data in order to explore its potential quantum advantages, particularly leveraging the residual design to enhance learning efficiency and generalization across structured quantum datasets.

Assume Res-HQCNN has $L$ hidden layers. Each quantum perceptron represents single qubit. The number of nodes in layer $l$ is denoted as $m_l$. In this paper, we assume $m_{l-1}\leq m_l$ for $l=1,2,\cdots,L+1$. $l=0$ and $l=L+1$ represent the input and output layers respectively. It is also important to emphasize that the output state of layer $l$ is not the real input of layer $l+1$ in Res-HQCNN. According to the residual block structure in Fig.\ref{Figure1}(b), the right input state of layer $l+1$ is an addition of the output and input state of layer $l$. In Fig.\ref{figure2}, we present the model of Res-HQCNN with graph-structured quantum data. 

\begin{figure*}[ht]
\centering
\includegraphics[angle=0,width=1\linewidth]{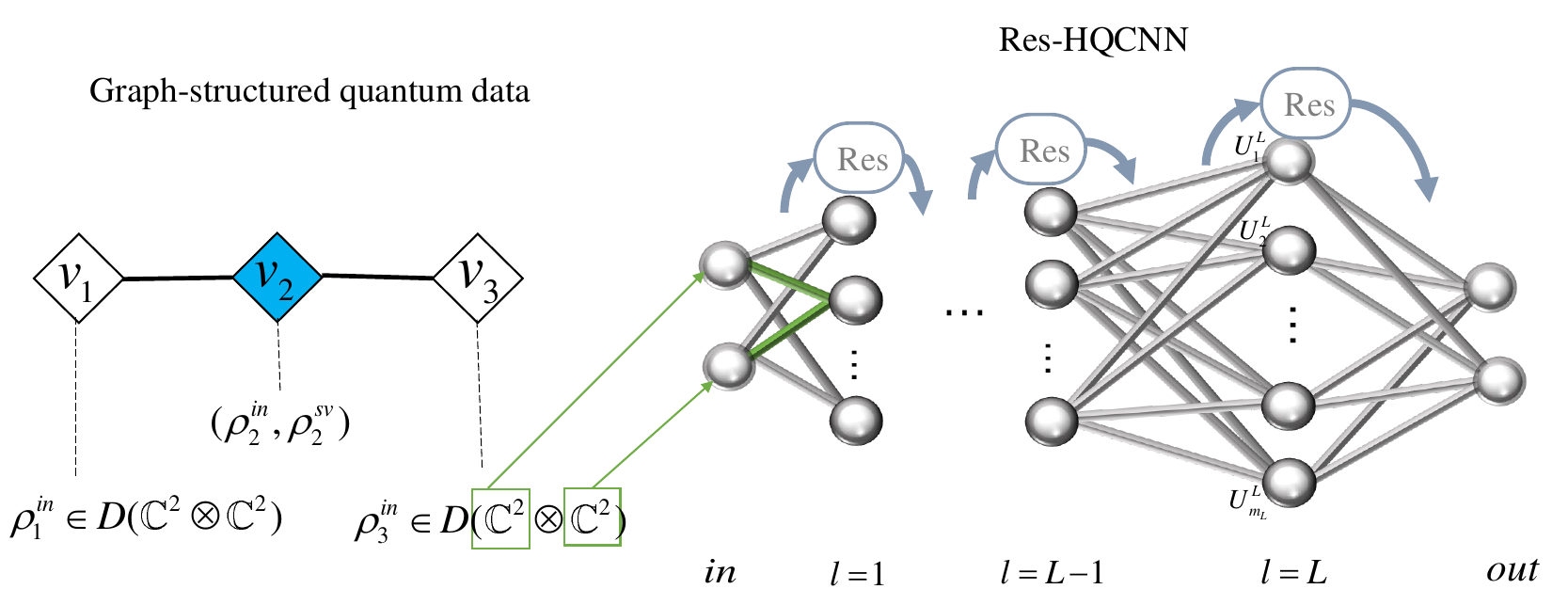}
\caption{\leftskip=0pt \rightskip=0pt plus 0cm  \textbf{Res-HQCNN and graph-structured quantum data}. ``Res'' represents the residual block structure in quantum concept in Fig. \ref{Figure1}(b). The green arrows mean that the dimension of the
input state $\rho_3^{in}$ ($\rho_1^{in}$) should match the nodes of the input layer in Res-HQCNN. The shaded vertex represents the supervised nodes with ideal output state $\rho_2^{sv}$.}
\label{figure2}
\end{figure*}

One may notice that the using of residual block structure in 
Fig.\ref{figure2} increases the trace value of the input state $\rho^{2_{in}}$. The trace value is not $1$, which means $\rho^{2_{in}}$ is not a density matrix at all. In other words, if we regard Res-HQCNN as a map $\mathcal{R}$, then $\mathcal{R}$ is not a completely positive map. However, in the experiment part, we will show that deeper Res-HQCNN can still learn graph-structured quantum data better than the state-of-the-art models.

\subsection{Cost function}

In order to investigate the training and testing performance of Res-HQCNN, we need to find suitable cost function for this special graph-structured quantum data.
Similar to the definition of cost function for graph-structured quantum data in Beer's model \cite{beer2023quantum}, we now define the cost function of Res-HQCNN for graph-structured quantum data in this paper. For convenience, we denote a map $\mathcal{R}$ as the acting of Res-HQCNN.

\subsubsection{Cost function for $S$ supervised part}

The first $S$ quantum states are supervised, meaning their corresponding labels are known. For these states, fidelity between the output state and the target state is commonly used as an information metric. Accordingly, the supervised component of the cost function for Res-HQCNN is defined as:
\begin{align}
C_{SV}=\frac{1}{2^tS}\sum_{x=1}^{S}\bra{\phi_{x}^{sv}}\mathcal{R}(\rho_{x}^{in})\ket{\phi_{x}^{sv}},
\end{align}
where $t$ is the number of residual block structure in a Res-HQCNN.

\subsubsection{Cost function for graph-based self-supervised part}

To capture more structural information from the graph $G$ and provide a more faithful embedding, the adjacency matrix $A$ can be incorporated into the cost function. Thus, the graph-based unsupervised component of the cost function can be defined as:
\begin{align}\label{cg}
  C_{G}=\frac{1}{2^t}\sum_{v,w\in V}[A]_{v,w}D_{HS}(\mathcal{R}(\rho_v),\mathcal{R}(\sigma_w)).  
\end{align}
Here $D_{HS}(\rho,\sigma)=\trace[(\rho-\sigma)^2]$ is the the Hilbert-Schmidt distance \cite{HSozawa2000entanglement},
$[A]_{v,w}$ is the matrix element of adjacency matrix $A$ corresponding to the vertices $v$ and $w$, and $t$ is the number of residual block structure in a Res-HQCNN. During the learning process of Res-HQCNN, we will minimize the cost function $C_G$ when $\mathcal{R}(\rho_v) $and $\mathcal{R}(\sigma_w)$ are information-theoretic close.

\subsubsection{The full cost function for training}

The full cost function can be defined as a combination of supervised part and graph-based self-supervised part:
\begin{align}\label{fullcost}
 C_{SV+G}=C_{SV}+\gamma C_{G},
\end{align}
where $\gamma$ is called as graph part control factor with $\gamma \leq 0$. 
Therefore, throughout the training process, we aim to maximize $C_{SV+G}$
to effectively exploit the structural information encoded in the graph 
$G$.

\subsubsection{The cost function for testing}

After training Res-HQCNN on graph-structured data using Eq.(\ref{fullcost}), we evaluate its performance by computing the testing loss, which is given by:
\begin{align}\label{tesingloss}
C_{USV}=\frac{1}{2^t(N-S)}\sum_{x=S+1}^{N}\bra{\phi_{x}^{test}}\mathcal{R}(\rho_{x}^{in})\ket{\phi_{x}^{test}}.
\end{align}
where $t$ is the number of residual block structure in a Res-HQCNN.

\subsection{The training algorithm}

Based on the theoretical and architectural preparations above, we now proceed to describe the training algorithm for Res-HQCNN in the context of graph-structured quantum data. Given that the training data may include both supervised samples with and without graph structure, we present the corresponding training procedures separately, detailing the cases with graph-based regularization and those without.

\subsubsection{Update unitary matrix for supervised training data without graph}

Since this case has been previously investigated in detail in Ref.\cite{liang2021hybrid}, we now briefly summarize the key expressions for the Res-HQCNN architecture in the case of a single hidden layer.

The unitaries matrix of Res-HQCNN can be updated via
\begin{align}\label{updateunitary}
 U_{j}^{l}(s+\epsilon)=e^{i\epsilon K_{j}^{l}(s)}U_{j}^{l}(s).
\end{align}
in which $s$ is the step and usually updated with $s=s+\epsilon$, $K_{j}^{l}$ is the parameters matrix. 

When $l=1$, we can compute an analytical expression of $K_{j}^{l}$ in Eq.{(\ref{updateunitary})},
\begin{align}\label{k1}
K_{j}^{l}(s)=\eta \frac{2^{m_{l-1}}}{S}\sum_{x=1}^{S}\trace_{rest}M_{j}^{l}.
\end{align}
Here the trace is over all qubits of Res-HQCNN not affected by $U_{j}^{l}$, and $\eta$ is the learning rate, $S$ is the number of supervised training data.
The matrix $M_{j}^{l}$ in Eq.{(\ref{k1})} is made up of two parts of the commutator:
\begin{align}\label{M1}
 M_{j}^{l}(s)=&[U_{j}^{l}(s)\cdots U_{1}^{l}(s)\left(\rho_{x}^{l_{in}}(s)\otimes \ket{0\cdots0}_{l}\bra{0\cdots0}\right)
 {U_{1}^{l}}^{\dagger}(s)\cdots 
{U_{j}^{l}}^{\dagger}(s),\nonumber\\
&{U_{j+1}^{l}}^{\dagger}(s)\cdots {U_{m_{out}}^{out}}^{\dagger}(s)\left(Id(m_{l-1})\otimes \ket{\phi_{x}^{out}}\bra{\phi_{x}^{out}}\right)
U_{m_{out}}^{out}(s)\cdots U_{j+1}^{l}(s)].
\end{align}

When $l=2$ (the output layer), the parameter matrix $K_{j}^{l}$ in Eq.{(\ref{updateunitary})} can be computed as 
\begin{align}\label{k2}
K_{j}^{l}(s)=\eta \frac{2^{m_{l-1}}}{S}\sum_{x=1}^{S}\trace_{rest}(M_{j}^{l}+N_{j}^{l}),
\end{align}
where $M_{j}^{l}$ is in Eq.(\ref{M1}) and
\begin{align}\label{N2}
 N_{j}^{2}(s)=&[U_{j}^{2}(s)\cdots U_{1}^{2}(s)(\rho_{x}^{in}(s)\otimes \ket{0\cdots0}_{\Delta m_1}\bra{0\cdots0}
\otimes\ket{0\cdots0}_{2}\bra{0\cdots0})
 {U_{1}^{2}}^{\dagger}(s)\cdots {U_{j}^{2}}^{\dagger}(s),\nonumber\\
&{U_{j+1}^{2}}^{\dagger}(s)\cdots {U_{m_{2}}^{2}}^{\dagger}(s)
\left(Id(m_{1})\otimes \ket{\phi_{x}^{out}}\bra{\phi_{x}^{out}}\right)
U_{m_{2}}^{2}(s)\cdots U_{j+1}^{2}(s)].
\end{align}

If Res-HQCNN has more than two hidden layers,
we can calculate its update parameters matrices $K_{j}^{l}(s)$ with similar computational method in Ref.\cite{liang2021hybrid}. No fixed formula for $K_{j}^{l}(s)$ has been found up to now. It changes with the depth of neural network $l$.

\subsubsection{Update unitary matrix for graph-structured quantum data}

Based on Eq.{(\ref{updateunitary})}, we first compute the parameter matrix $K_{j}^{l}(s)$ for graph-structured quantum data in Res-HQCNN with one hidden layer. 

\begin{theorem}\label{thm1}
 The update matrix for a Res-HQCNN with one hidden layer trained with a graph structure between output states ${\rho_v^{out},\rho_w^{out}}$  encoded with a adjacency matrix $[A]_{vw}$ is
\begin{align}\label{knew1}
K_{j}^{l}(s)=&\eta {2^{m_{l-1}+1}}i\sum_{v\sim w}\trace_{rest}(M_{j\{v,w\}}^{l}(s)),\qquad l=1,\ j=1,2,\cdots,m_l;
\end{align}
\begin{align}\label{knew2}
K_{j}^{l}(s)=&\eta {2^{m_{l-1}+1}}i\sum_{v\sim w}
\trace_{rest}(M_{j\{v,w\}}^{l}(s)+
N_{j\{v,w\}}^{l}(s)),\ 
l=2,\ j=1,2,\cdots,m_l.
\end{align}
where
\begin{align}\label{Mnew1}
 M_{j\{v,w\}}^{l}(s)=&[U_{j}^{l}(s)\cdots U_{1}^{l}(s)((\rho_{v}^{l_{in}}-\rho_{w}^{l_{in}})(s)\otimes 
 \ket{0\cdots0}_{l}\bra{0\cdots0}){U_{1}^{l}}^{\dagger}(s)\cdots 
{U_{j}^{l}}^{\dagger}(s),\nonumber\\
&{U_{j+1}^{l}}^{\dagger}(s)\cdots {U_{m_{out}}^{out}}^{\dagger}(s)
(Id(m_{l-1})\otimes 
(\rho_{v}^{out}-\rho_{w}^{out}))
U_{m_{out}}^{out}(s)\cdots U_{j+1}^{l}(s)].
\end{align}
\begin{align}\label{Nnew1}
 N_{j}^{2}(s)=&[U_{j}^{2}(s)\cdots U_{1}^{2}(s)((\rho_{v}^{in}-\rho_{w}^{in})(s)\otimes 
 \ket{0\cdots0}_{\Delta m_1}\bra{0\cdots0}\otimes\ket{0\cdots0}_{2}\bra{0\cdots0})\nonumber\\
 &{U_{1}^{2}}^{\dagger}(s)\cdots {U_{j}^{2}}^{\dagger}(s),
{U_{j+1}^{2}}^{\dagger}(s)\cdots {U_{m_{2}}^{2}}^{\dagger}(s)
\left(Id(m_{1})\otimes(\rho_{v}^{out}-\rho_{w}^{out})\right)
U_{m_{2}}^{2}(s)\cdots U_{j+1}^{2}(s)].
\end{align}
\end{theorem}

The proof of Theorem \ref{thm1}, provided in the Appendix, begins with the definition of the derivative function and utilizes several properties of the trace operator. When the number of hidden layers in Res-HQCNN exceeds two, the corresponding parameter update matrices can be derived using a similar computational approach as outlined in the proof of Theorem \ref{thm1}.

\subsubsection{Full update matrix}

Based on the mathematical expressions derived above, we can get the full update matrix for supervised vertices with or without graph. For three-layer Res-HQCNN, combining Eq.(\ref{k1}) with Eq.(\ref{knew1}), we obtain the update parameter matrix $K_{j}^{l}(s)$ for the training data in Eq.(\ref{data}) of Res-HQCNN with one hidden layer when $l=1$:
\begin{align}\label{ktotal}
K_{j}^{l}(s)=&\eta \frac{2^{m_{l-1}}i}{S}\sum_{x=1}^{S}\trace_{rest}M_{j\{x\}}^{l}
+\gamma [\eta {2^{m_{l-1}+1}}i\sum_{v\sim w}\trace_{rest}(M_{j\{v,w\}}^{l}(s))], \nonumber\\
&l=1,\ j=1,2,\cdots,m_l;
\end{align}
Here graph part control factor $\gamma \leq 0$.

Before presenting the following experiments, a deeper theoretical perspective may help elucidate the underlying mechanisms contributing to the performance improvement of Res-HQCNN. We outline two key theoretical components: the role of residual connections in enhancing trainability, and the influence of graph-based information on generalization. In Res-HQCNN, each layer’s output is defined as
\begin{align}\label{add1}
\rho^{l_{out}}=U^{l}\rho^{l_{in}}{U^{l}}^\dagger +\rho^{l_{in}},
\end{align}
in which $U^{l}=U_{m_l}^{l}\cdots U_1^{l}$ represents the quantum perceptrons of the layer $l$ in Fig.\ref{figure2}. 
This additive residual structure introduces an identity path through which gradient information can flow directly, improving convergence in deeper networks.
These insights align with recent work on quantum residual networks \cite{kashif2024resqnets,wen2024enhancing}, which suggests that residual blocks help avoid overparameterization while maintaining the model’s capacity to approximate general unitaries.
The second source of improvement in Res-HQCNN lies in its use of graph-based regularization. The graph-structured loss function $C_G$ in Eq.(\ref{cg}) encourages the output states $\mathcal{R}(\rho_v)$ and $\mathcal{R}(\rho_w)$ of adjacent nodes $v,w \in V$ to remain close in Hilbert-Schmidt distance. This acts as an inductive bias that enforces local smoothness over the graph structure, constraining the hypothesis space of the quantum model and guiding it toward solutions that respect known correlations in the data.
Together, residual learning and graph-based regularization work synergistically. Residual blocks make the model deeper and more trainable; graph information stabilizes the learning process by incorporating known structural correlations. This dual design enables Res-HQCNN to achieve faster convergence, higher fidelity, and improved robustness theoretically.
A complete theoretical formalism that fully quantifies this synergy remains an open problem, but our analysis provides a conceptual framework that helps explain the model’s observed performance advantages.

\section{Experimental Methods and results}\label{5}

In this section, we conduct a comprehensive set of experiments to evaluate the learnability of Res-HQCNN on graph-structured quantum data. All experiments are performed on a FusionServer Pro G5500. The training data used follow the format specified in Eq.(\ref{data}). Aligned with the two research objectives outlined earlier, we first design comparative experiments to examine the effectiveness of incorporating graph structure information. Subsequently, we increase the depth of the Res-HQCNN model to investigate whether residual learning enhances its ability to learn from graph-structured quantum data.

Since the QNN model proposed in 2020 is both universal and efficient, exhibiting no barren plateau phenomenon \cite{beer2020training}, and Res-HQCNN is developed based on this universal QNN architecture, the primary structural difference lies in the addition of residual connections.
In contrast, other hybrid quantum-classical models, quantum graph neural networks, or quantum convolutional neural networks differ significantly in architecture and design principles from Res-HQCNN.
Therefore, using the universal QNN in Ref.\cite{beer2020training} as a baseline is well-justified for evaluating the impact of residual connections and the incorporation of graph information.
In the following figures, the solid lines represent the results from  Res-HQCNN model in this paper while the dashed ones corresponding to the universal QNN model.
For convenience, we apply a 1-dimensional list of natural numbers to refer to the number of perceptrons in the corresponding layer. 
If a residual block structure acts on the hidden layers of Res-HQCNN, we add a tilde on top of the
natural numbers. For example, a 1-2-1 QNN in Ref.\cite{beer2023quantum} can be denoted as $[1,2,1]$, and a 1-2-1 Res-HQCNN in this paper is denoted as $[1,\tilde{2},1]$.  

We will use the training and testing losses to numerically evaluate the performance of Res-HQCNN model. Since the input and output of Res-HQCNN are quantum states, the fidelity and the Hilbert-Schmidt distance used in the definition of cost function are suitable for the specific structure of Res-HQCNN, just like the key performance indicators (KPIs) of mean absolute percentage error (MAPE), mean squared error (MSE), mean absolute error (MAE) in classical regression predictions.

Moreover, we adopt fidelity-based cost functions, which measure how close the output quantum state to the target state. Since fidelity lies in $[0,1]$, higher values indicate better learning performance. Consequently, plots of training and testing loss in this paper represent fidelity scores, not conventional loss functions. This choice aligns with common practice in state-of-the-art quantum machine learning research \cite{beer2023quantum,liang2021hybrid,beer2020training}, but may differ from classical machine learning terminology where lower loss typically indicates improvement.

\begin{figure*}[htbp]
\centering
\subfigure[Line]{
\includegraphics[width=5cm]{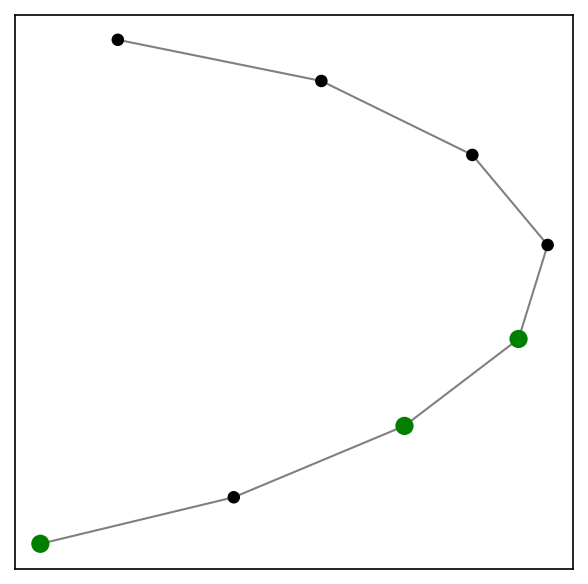}
}
\subfigure[Connected clusters]{
\includegraphics[width=5cm]{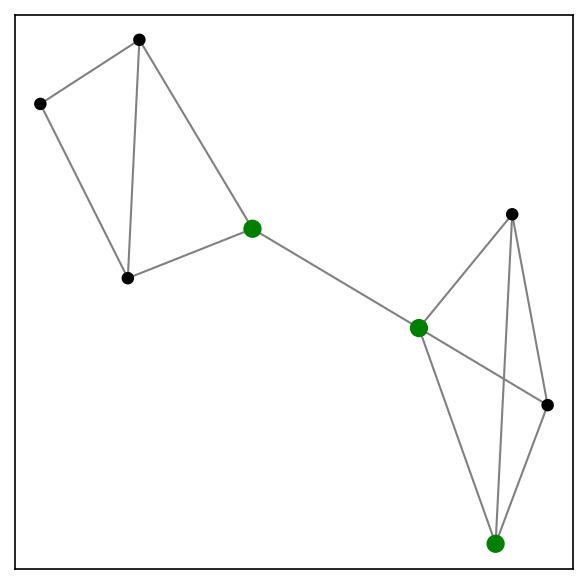}
}
\caption{ \textbf{The output states of training data $(\rho_{v_j},p_{v_j})$ in Eq.(\ref{data}) with $j=1,2,\cdots,8$ comprise a graph with clusters or line. The green shades vertices are used for training}}
\label{Figure3}
\end{figure*} 

\begin{figure*}[htbp]
\centering
\subfigure[Line]{
\includegraphics[width=6cm]{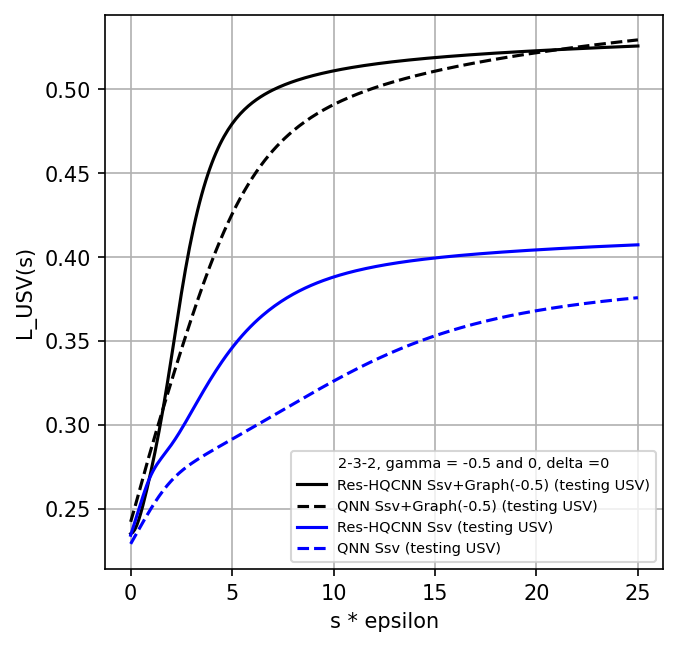}
}
\subfigure[Connected clusters]{
\includegraphics[width=6cm]{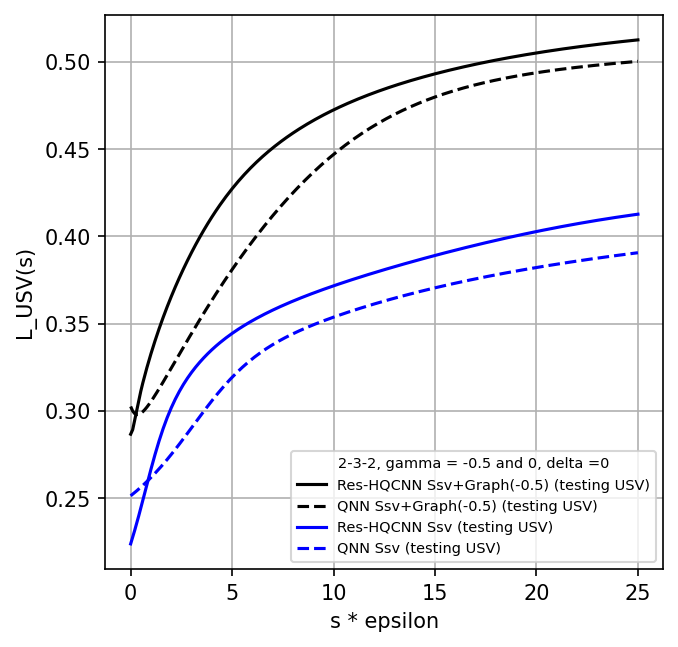}
}
\caption{ \textbf{Numerical results of $[2,\tilde{3},2]$(solid) and $[2,3,2]$(dashed) for $8$ training data with $3$ supervised}. Testing loss during 250 training epochs. The black lines are with graph($\gamma=-0.5$), and the blue lines are without graph($\gamma=0$). (a)  The $3$ supervised data corresponds to green circles of line in Fig.\ref{Figure3}(a); (b) The $3$ supervised data corresponds to green circles of connected clusters in Fig.\ref{Figure3}(b).}
\label{Figure4}
\end{figure*}

In fact, $N$ quantum training data can comprise many different kinds of graphs, such as circles, lines and connected clusters. In this experiments, we randomly choose the number of training pairs as $8$, and choose the graphs in line and connected clusters, which is shown in Fig.\ref{Figure3}. The vertices shaded in green are used for training, and the rest are used for testing.

\subsection{Elementary test with one hidden layer}

As an elementary test, we choose Res-HQCNN $[2,\tilde{3},2]$ and the QNN $[2,3,2]$ to learn graph-structured data. 
The testing loss values in Eq.(\ref{tesingloss}) for $[2,\tilde{3},2]$ and $[2,3,2]$ are shown in Fig.\ref{Figure4} in terms of line and connected clusters in Fig.\ref{Figure3}.
Here $\gamma$ is a negative number, which is the weight of the graph part in the full training cost. Here set the step size $\epsilon=0.01$.

In Fig.\ref{Figure4}, when the number of training pairs is $3$ and $\gamma=-0.5$, all the black curves are consistently above the blue ones, indicating that incorporating graph information enhances the learning performance of the networks. Recall that the solid lines correspond to the Res-HQCNN architecture $[2,\tilde{3},2]$, and the dashed lines are the ones from the universal QNN $[2,3,2]$. All the solid lines are higher than the dashed ones, which reflects the residual approach can help the network learn the information of graph better. Additionally, by comparing Fig.\ref{Figure4} (a) and (b), we observe that both models $[2,\tilde{3},2]$ and $[2,3,2]$ 
converge more quickly on line-structured data than on clustered graphs. Moreover, the gap between the black and blue lines is more pronounced in the line structure, suggesting that training data with line graph topology may be better suited for learning with Res-HQCNN.

When the number of training pairs varies, the results reveal interesting trends. As shown in Fig.\ref{Figure5}, we examine how performance is affected by the number of supervised quantum states. For line-structured data, all black data points are consistently higher than the brown ones across different numbers of supervised states. This trend aligns with the behavior of the black solid and dashed lines in Fig.\ref{Figure5}, further confirming the advantage of incorporating graph information. But for connected clusters, when the number of supervised states is $1$, $4$, $6$ and $7$, the brown points are higher. Based on this generalization analysis, we conclude that the Res-HQCNN architecture with a single hidden layer appears to be better suited for datasets with line-like graph structures, rather than those with connected-cluster topologies.

\begin{figure*}[htbp]
\centering
\subfigure[Line]{
\includegraphics[width=6cm]{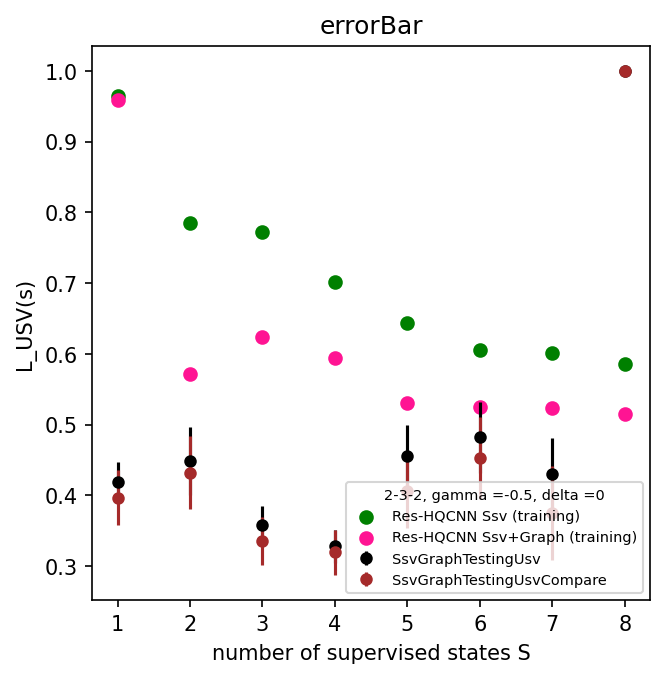}
}
\subfigure[Connected clusters]{
\includegraphics[width=6cm]{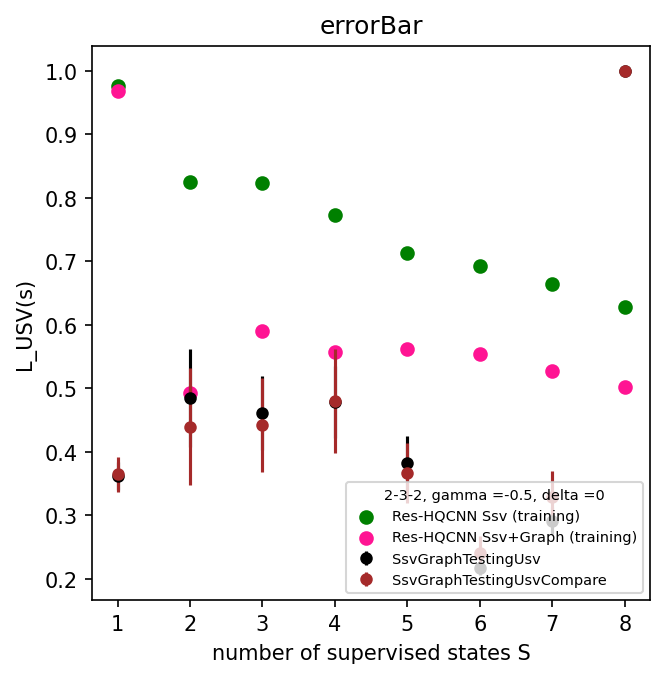}
}
\caption{ \textbf{The training and testing value of different number of supervised states $S$ of $[2,\tilde{3},2]$ and $[2,3,2]$ with graph.} The black data points are the results of Res-HQCNN with graph for testing, while the brown data points are the ones for universal QNN with graph. Each data point is an average over $8$ independent training attempts, and the error bars are one standard error of the mean. Here we also record the training values for Res-HQCNN with or without graph in green and pink data points.}
\label{Figure5}
\end{figure*} 

\begin{figure*}[htbp]
\centering
\subfigure[Line]{
\includegraphics[width=6cm]{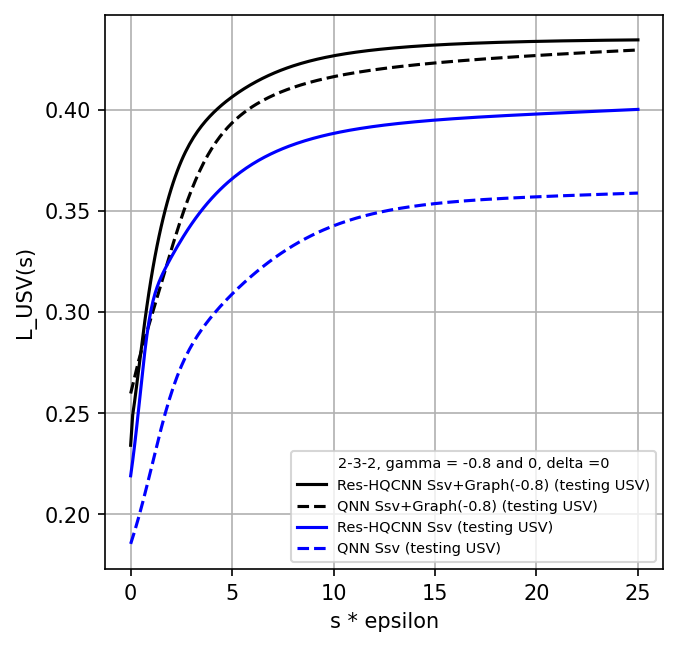}
}
\subfigure[Error bar]{
\includegraphics[width=6cm]{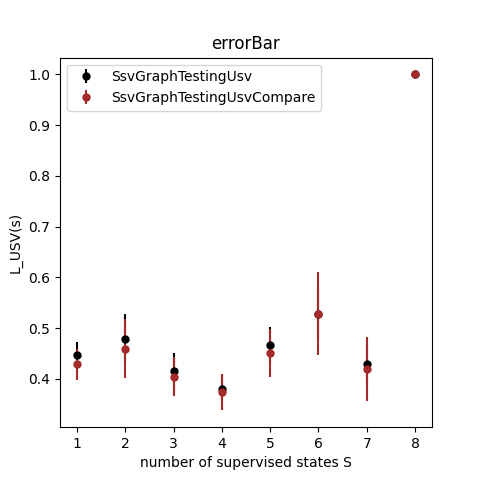}
}
\caption{ \textbf{Numerical results of $[2,\tilde{3},2]$(solid) and $[2,3,2]$(dashed) for $8$ training data with $3$ supervised}. Testing loss during 250 training epochs. The black lines are with graph($\gamma=-0.8$), and the blue lines are without graph($\gamma=0$). }
\label{Figure6}
\end{figure*} 

Here we also report the training results of Res-HQCNN with or without graph in Fig.\ref{Figure5}. As the number of supervised states increases, the training values represented by the green and pink points decrease, which is expected due to the definition of $C_{SV}$. Since the graph part control factor $\gamma$ is negative in the mathematical expression of Eq.(\ref{fullcost}), the pink data points are always lower than the green ones. However, from these training values alone, it is difficult to draw meaningful conclusions related to our two primary objectives. Therefore, in the following sections, we focus on the testing results to better evaluate the impact of graph information and residual connections on the model’s generalization performance.

What is the effect of changing the graph weight parameter $\gamma$? We present the results for $\gamma=-0.8$ in Fig.\ref{Figure6}. When using $8$ training samples with $3$ supervised states, the convergence value does not exceed $0.45$ in Fig.\ref{Figure6}(a), which is smaller than the one in Fig.\ref{Figure4}(a). Moreover, as shown by the error bars in Fig.\ref{Figure6}(b), the difference between the black and brown points is also smaller than those in Fig.\ref{Figure5}(a) with $\gamma=-0.5$. So $-0.5$ appears to be a more appropriate choice for $\gamma$ to effectively evaluate the strengths and weaknesses of Res-HQCNN with one hidden layer on line-structured datasets.

\subsection{The network with two hidden layers}

\begin{figure*}[htbp]
\centering
\subfigure[Line]{
\includegraphics[width=6cm]{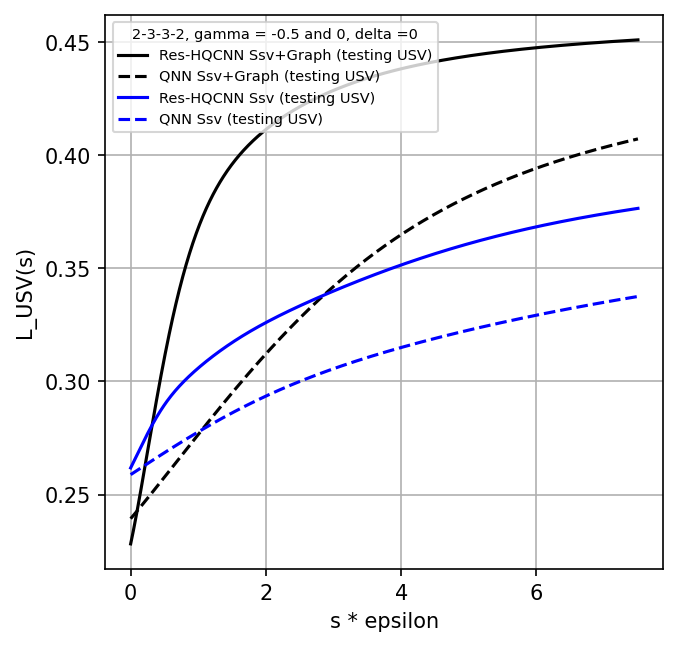}
}
\subfigure[Connected clusters]{
\includegraphics[width=6cm]{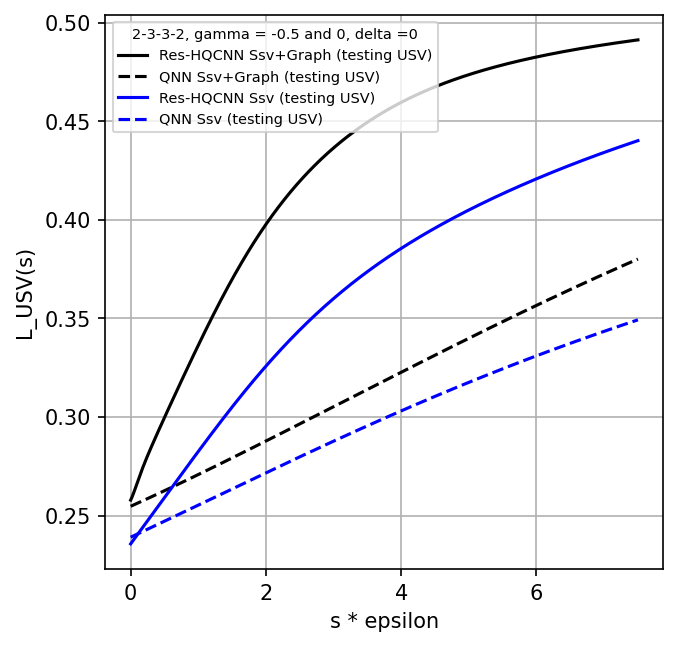}
}
\subfigure[Line]{
\includegraphics[width=6cm]{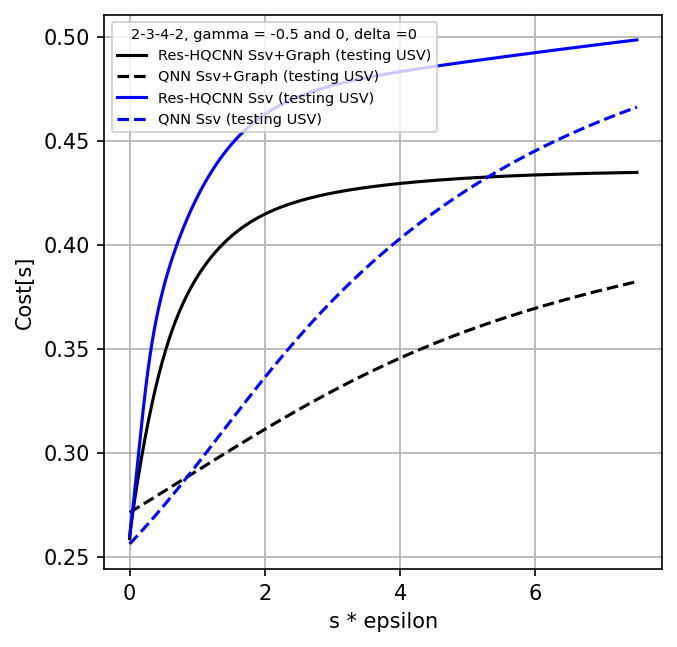}
}
\subfigure[Connected clusters]{
\includegraphics[width=6cm]{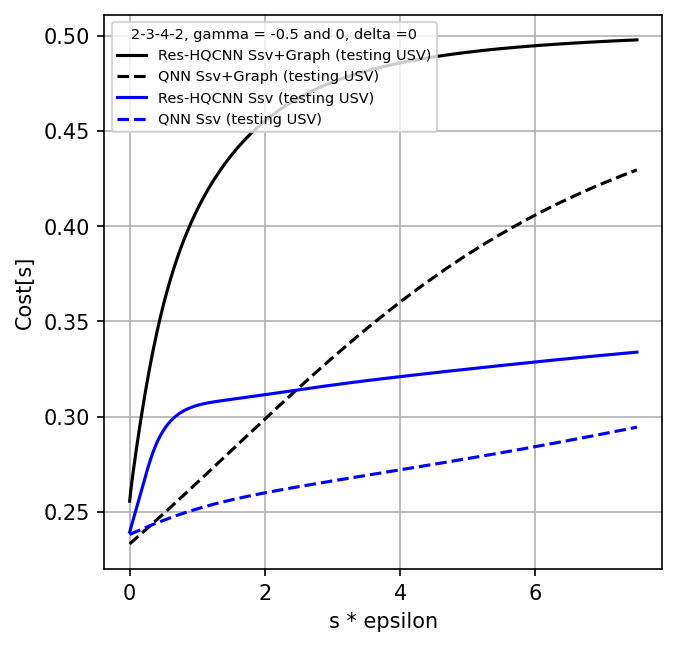}
}
\caption{ \textbf{Numerical results of Res-HQCNN with two hidden layers for $8$ training data with $3$ supervised in $750$ training epochs}.  The black lines are with graph($\gamma=-0.5$), while the blue lines are without graph($\gamma=0$).}
\label{Figure7}
\end{figure*} 

In Fig.\ref{Figure7}(a) and (b), we present the testing loss for the architectures $[2,\tilde{3},\tilde{3},2]$ and $[2,3,3,2]$ using $8$ training data with $3$ supervised. The incorporation of residual blocks in two hidden layers leads to improved performance on both line-structured and connected-cluster data, as evidenced by the comparison between the solid and dashed lines. Additionally, the inclusion of graph information in the input data results in better learning outcomes, as shown by the higher performance of the black lines compared to the blue lines.

In Fig.\ref{Figure7}(c) and (d), we also conduct the experiment for $[2,\tilde{3},\tilde{4},2]$ and $[2,3,4,2]$.
In Fig.\ref{Figure7}(c), with $8$ training samples and $3$ supervised states, the Res-HQCNN $[2,\tilde{3},\tilde{4},2]$ with graph information does not achieve the good performance. This may be attributed to the increased number of quantum perceptrons, which adds complexity, as well as to the randomness in the input training data. As shown in Fig.\ref{Figure3}, the data are randomly sampled from a normal distribution in each trial, and in some cases, this randomness may lead to suboptimal results, such as those seen in Fig.\ref{Figure7}(c). However, when the number of supervised samples is increased to $6$, as shown in Fig.\ref{Figure7}(d), $[2,\tilde{3},\tilde{4},2]$ converges more quickly and reaches a slightly higher final value. These results suggest that when modifying the architecture of Res-HQCNN, it is important to carefully tune the associated parameters in order to achieve optimal performance.

We also show the testing value of different number of supervised states $S$ of $[2,\tilde{3},\tilde{3},2]$ and $[2,3,3,2]$ in Fig.\ref{Figure8}.
We observe that for the network with two hidden layers, all the black data points consistently outperform the brown ones. Moreover, the gap between the black and brown points is larger than that seen in Fig.\ref{Figure5}. This indicates that the Res-HQCNN architecture $[2,\tilde{3},\tilde{3},2]$ demonstrates greater stability and more effectively leverages graph information compared to the shallower model $[2,\tilde{3},2]$. The residual block structure gives a more pronounced impact in deeper networks, which is consistent with the theoretical analysis presented earlier.
\begin{figure*}[htbp]
\centering
\subfigure[Line]{
\includegraphics[width=6cm]{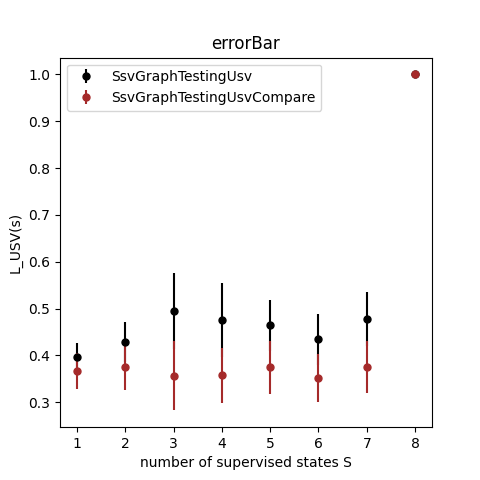}
}
\subfigure[Connected]{
\includegraphics[width=6cm]{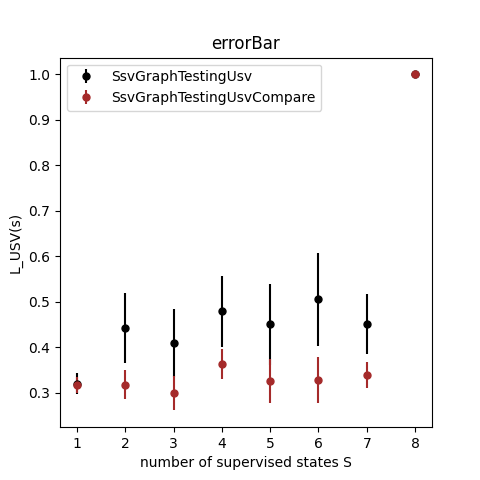}
}
\caption{ \textbf{The testing value of different number of supervised states $S$ of $[2,\tilde{3},\tilde{3},2]$ and $[2,3,3,2]$ with graph.} }
\label{Figure8}
\end{figure*} 

\begin{figure*}[htbp]
\centering
\subfigure[Line]{
\includegraphics[width=6cm]{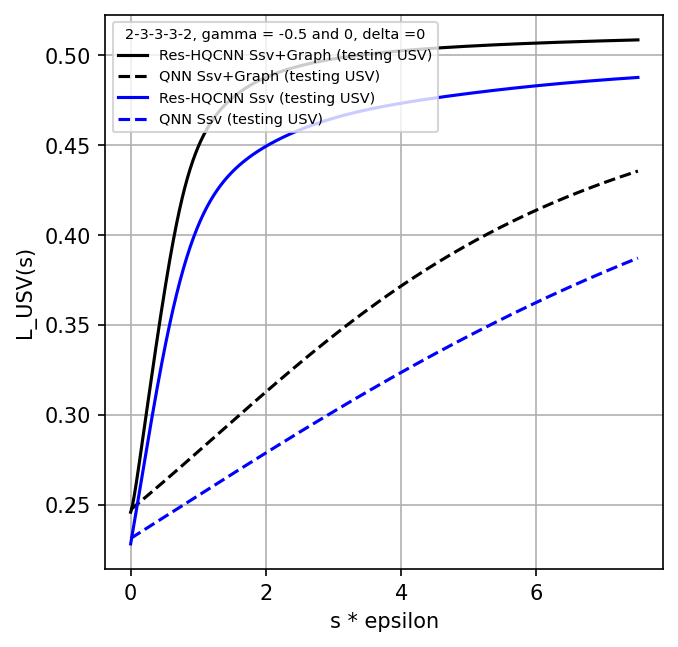}
}
\subfigure[Connected clusters]{
\includegraphics[width=6cm]{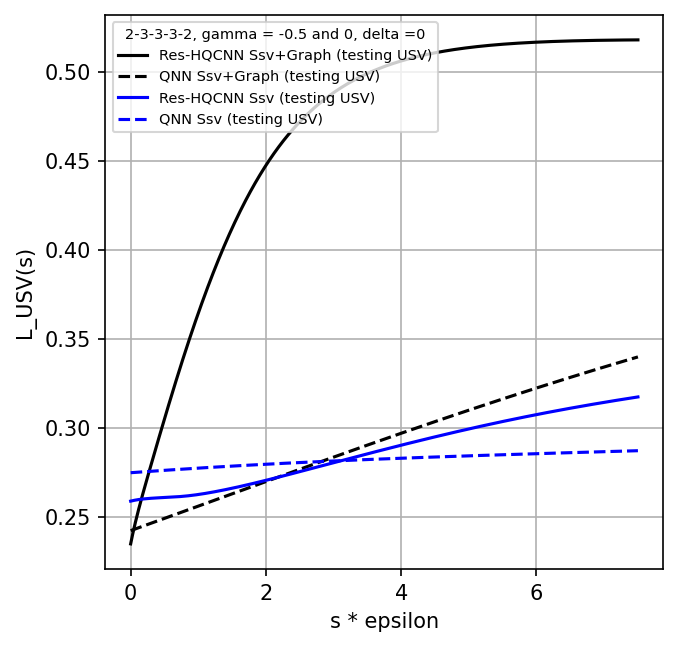}
}
\subfigure[Line]{
\includegraphics[width=6cm]{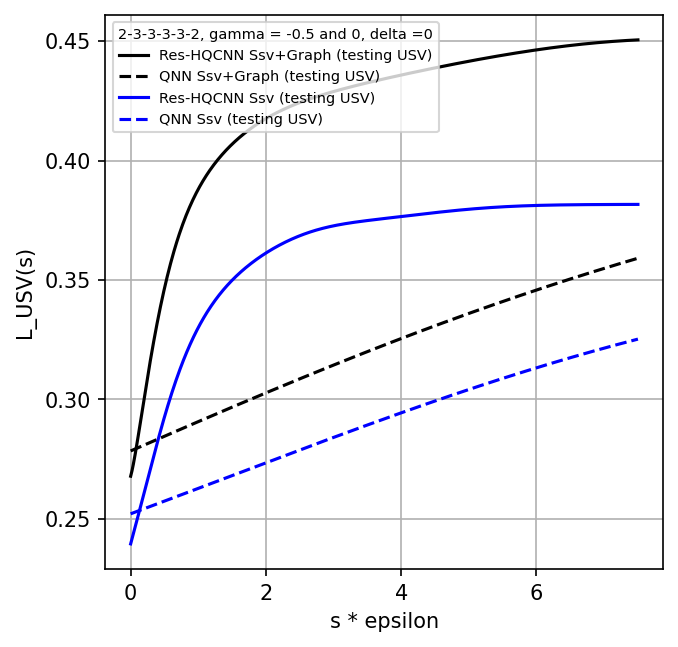}
}
\subfigure[Connected clusters]{
\includegraphics[width=6cm]{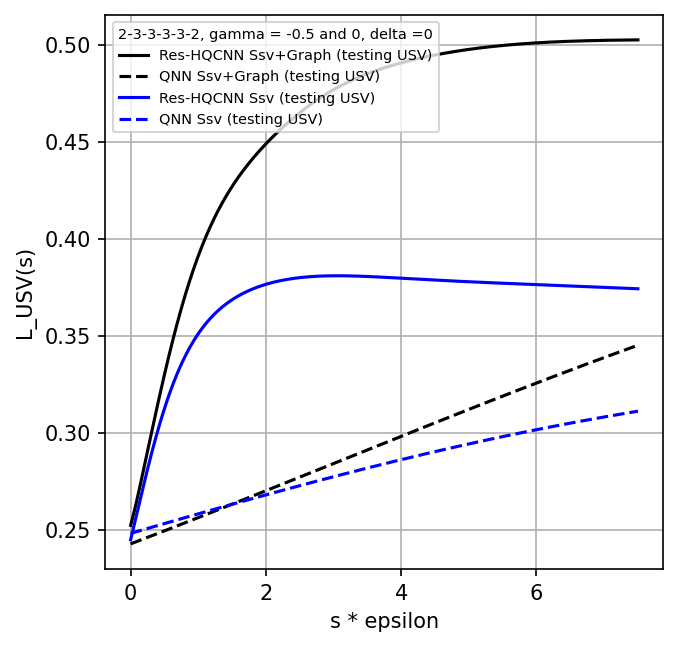}
}
\subfigure[Line]{
\includegraphics[width=6cm]{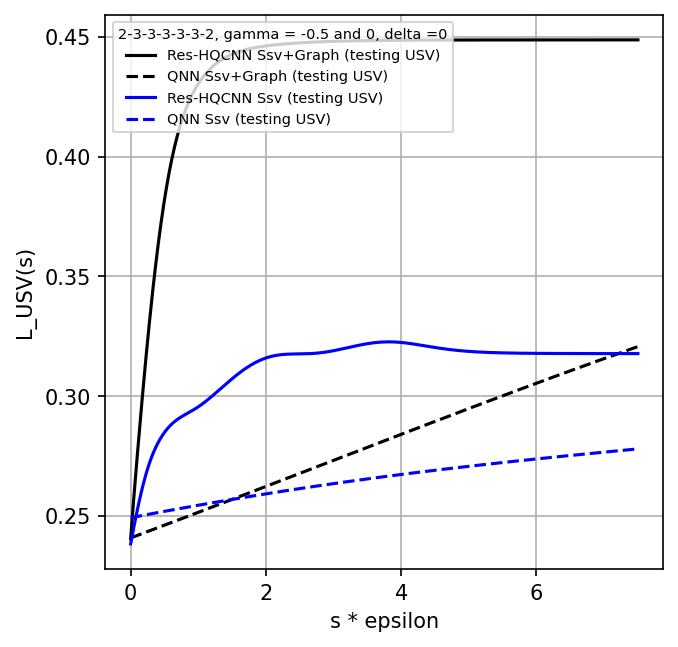}
}
\subfigure[Connected clusters]{
\includegraphics[width=6cm]{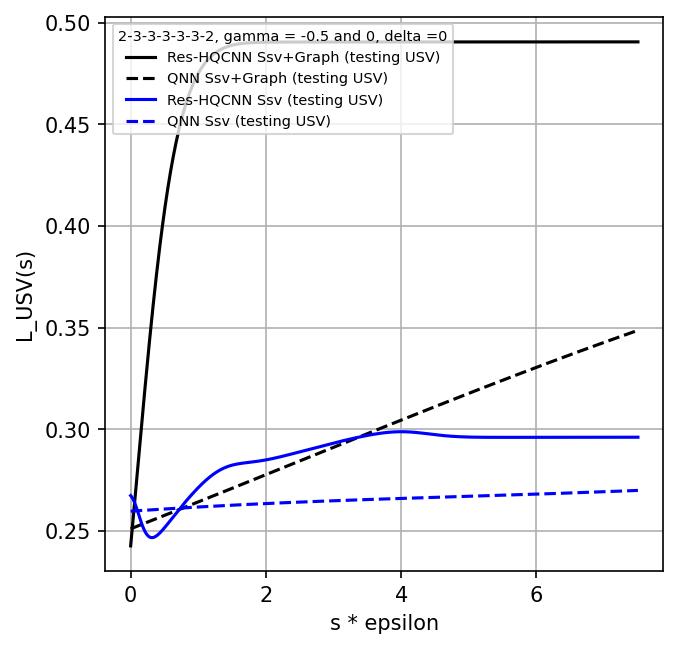}
}
\caption{ \textbf{Numerical results of Res-HQCNN with more than three hidden layers for $8$ training data with $4$ supervised in $750$ training epochs}. The black lines are with graph($\gamma=-0.5$), and the blue lines are without graph($\gamma=0$).}
\label{Figure10}
\end{figure*}

\subsection{The network with more than three hidden layers}
In this part, we investigate whether deeper networks lead to better performance. Specifically, we examine Res-HQCNN with $3$, $4$ and $5$ hidden layers for graph-structured quantum data as an example in Fig.\ref{Figure10}. 
The comparison between the black and blue lines shows that incorporating graph information into the training data improves test performance for both Res-HQCNN and universal QNNs. While the residual block structure in Res-HQCNN enhances the network's ability to learn from graph-structured data, as evidenced by the superior performance of the solid and dashed lines. For example, for the network of  
$[2,\tilde{3},\tilde{3},\tilde{3},\tilde{3},\tilde{3},2]$ in Fig.\ref{Figure10}(f), the blue lines without graph show poor performance. Moreover, because the cost function's performance is influenced by the number of supervised pairs and the topology in graph structure as shown in Eq.(\ref{tesingloss}), using $8$ training samples with $4$ supervised pairs, as in Fig.\ref{Figure10}, leads to a poorer performance shown in Fig.\ref{Figure10}(b) compared to Fig.\ref{Figure10}(a). However, the black solid line still achieves the best performance in Fig.\ref{Figure10}(b), which remains consistent with the previous analysis.

Regarding computational complexity, the residual block structure in Res-HQCNN does not introduce additional overhead, as the element-wise addition is negligible. However, it remains difficult to precisely characterize the computational or time complexity of both Res-HQCNN and universal QNNs due to the many factors involved in training, such as the number of parameters, quantum circuit structure, optimization landscape, and the stochastic nature of gradient-based algorithms.
To provide empirical insight, we recorded the running times shown in Fig.\ref{Figure7}(a) and Fig.\ref{Figure10}(a). In Fig.\ref{Figure7}(a), the running times for the black solid, black dashed, blue solid, and blue dashed lines are $1152.03$, $600.16$, $222.8$ and $120.81$ seconds, respectively. In Fig.\ref{Figure10}(a), the corresponding times are $2539.15$, $884.64$, $547.35$ and $234.31$ seconds. These results indicate that both increased network depth and the use of graph information may lead to higher computational costs.
In future work, we plan to explore strategies to alleviate the computational burden, including more efficient quantum data encoding schemes, reducing trainable parameters via parameter sharing or sparse connectivity, and applying circuit simplification techniques to optimize circuit depth and width. Additionally, hybrid training approaches that combine classical preprocessing with quantum layers may offer improved scalability.

\section{Conclusion and discussion}\label{6}

In this paper, we build a Res-HQCNN model with graph-structured quantum data, and its corresponding training algorithm are explained in detail. Then, we conduct different kinds of comparable experiments to test the learning efficiency of Res-HQCNN for graph-structured data. Compared with the universal QNN model, we have shown that our Res-HQCNN model performs better to learn graph-structured quantum data. Although the information about graph and residual block structure are completely different concepts from math and deep learning, we find that the information about graph can bring better training and testing results no matter how many pairs are supervised. And at the same time, the residual block structure can help deep networks learn graph-structured data faster and better. 

We would like to clarify that this paper is fundamentally different from Ref.\cite{beer2023quantum} and Ref.\cite{liang2021hybrid}. 
Our work specifically focuses on graph-structured quantum data and introduces a tailored Res-HQCNN architecture to handle such data effectively. The update unitaries used in our training algorithm are newly designed, and the corresponding computational process is considerably more intricate. To the best of our knowledge, few studies have explored the interaction between residual learning and the graph structure of quantum data. These are two distinct concepts, yet both significantly influence the training dynamics and final model performance. Our work aims to bridge this gap by investigating their combined effect in a quantum learning framework.

We note that all quantum states, graphs, and labels used in our experiments are synthetically or randomly generated, with manually defined graph topologies such as line or cluster structures. This choice was made to create a controlled environment for evaluating the core capabilities of Res-HQCNN in handling graph-structured quantum data. 
However, we fully recognize that real quantum data, whether obtained through simulation or physical experiments, often do not exhibit explicit graph structures and are inherently affected by various types of noise. These factors pose significant challenges for practical deployment. Although our current study does not address the process of deriving graph structures from real quantum states, we consider this a critical direction for future research. In practice, constructing graph representations from quantum data may involve analyzing entanglement patterns, correlation measures, or spatial configurations of qubits. 
With respect to noise, we agree that robustness is essential for real-world applicability. Although our current experiments are conducted under idealized conditions, the residual architecture in Res-HQCNN is inherently well-suited for mitigating signal degradation \cite{liang2021hybrid}. In future work, we plan to introduce realistic noise models to systematically evaluate the robustness of the proposed framework.

\backmatter

\bmhead{Acknowledgements}
This work is supported by the Youth Doctoral "Starter" Project under Guangzhou Basic and Applied Basic Research Program (No. 2024A04J4450) and National Natural Science Foundation of China (No. 62471187). This work is also supported by the Key Lab of Guangzhou for Quantum Precision Measurement under Grant No. 202201000010, and the Key Research and Development Project of Guangdong Province under Grant No. 2020B0303300001.

\bmhead{Code Availability} The codes are available upon request from the authors.

\section*{Declarations}

\textbf{Conflict of interest} The authors declare that they have no known competing financial interests or personal relationships that could have appeared to influence the work reported in this paper.


\bibliography{sn-bibliography}

\begin{appendices}
\section{The proof of Theorem 1}
\begin{proof}
   The update matrix for a Res-HQCNN $[2,\Tilde{3},2]$ can be calculated in the following steps.  
We begin the process from the derivative function for cost function: $\frac{dC}{ds}=\lim\limits_{\epsilon \to 0}{\frac{C(s+\epsilon)-C(s)}{\epsilon}}$. Based on the definition of function $C_{G}(s+\epsilon)$ in Eq.(\ref{cg}), we firstly consider the output state of the updated unitary
\begin{align}\label{a1}
\setcounter{equation}{0}
\renewcommand\theequation{A.\arabic{equation}}
U^{1}(s+\epsilon)=e^{i\epsilon K_{3}^{1}(s)}U_{3}^{1}(s)e^{i\epsilon K_{2}^{1}(s)}U_{2}^{1}(s)e^{i\epsilon K_{1}^{1}(s)}U_{1}^{1}(s), 
\end{align}
\begin{align}\label{a2}
U^{2}(s+\epsilon)=e^{i\epsilon K_{2}^{2}(s)}U_{2}^{2}(s)e^{i\epsilon K_{1}^{2}(s)}U_{1}^{2}(s).
\end{align}
As for the output state of the updated unitary $U^{1}(s+\epsilon)$ in the hidden layer, for convenience, we omit to write the parameter $s$ in $K_{j}^{1}(s)$ and $U_{j}^{1}(s)$ with $j=1,2,3$. Then using Eq.(\ref{a1}), we have
\begin{align}\label{a3}
    \rho_x^{1_{out}}(s+\epsilon)=&\trace_{in}(e^{i\epsilon K_{3}^{1}}U_{3}^{1} e^{i\epsilon K_{2}^{1}}U_{2}^{1} e^{i\epsilon K_{1}^{1}}U_{1}^{1} (\ket{\phi_{x}^{in}}\bra{\phi_{x}^{in}}\otimes\ket{000}_{hid}\bra{000})
    \nonumber\\
    &{U_{1}^{1}}^{\dagger}e^{-i\epsilon K_{1}^{1}}
    {U_{2}^{1}}^{\dagger}e^{-i\epsilon K_{2}^{1}} 
    {U_{3}^{1}}^{\dagger}e^{-i\epsilon K_{3}^{1}})\nonumber\\
    =&\rho_x^{1_{out}}(s)+i\epsilon \trace_{in}(U_{3}^{1}U_{2}^{1}K_{1}^{1}U_{1}^{1}(\ket{\phi_{x}^{in}}\bra{\phi_{x}^{in}}\otimes\ket{000}_{hid}\bra{000})\nonumber\\
    &{U_{1}^{1}}^{\dagger}{U_{2}^{1}}^{\dagger}{U_{3}^{1}}^{\dagger}
+U_{3}^{1}K_{2}^{1}U_{2}^{1}U_{1}^{1}\left(\ket{\phi_{x}^{in}}\bra{\phi_{x}^{in}}\otimes\ket{000}_{hid}\bra{000}\right)\nonumber\\
&{U_{1}^{1}}^{\dagger}{U_{2}^{1}}^{\dagger}{U_{3}^{1}}^{\dagger}
+K_{3}^{1}U_{3}^{1}U_{2}^{1}U_{1}^{1}\left(\ket{\phi_{x}^{in}}\bra{\phi_{x}^{in}}\otimes\ket{000}_{hid}\bra{000}\right)
\nonumber\\
&{U_{1}^{1}}^{\dagger}{U_{2}^{1}}^{\dagger}{U_{3}^{1}}^{\dagger}-U_{3}^{1}U_{2}^{1}U_{1}^{1}\left(\ket{\phi_{x}^{in}}\bra{\phi_{x}^{in}}\otimes\ket{000}_{hid}\bra{000}\right){U_{1}^{1}}^{\dagger}{U_{2}^{1}}^{\dagger}{U_{3}^{1}}^{\dagger}K_{3}^{1}\nonumber\\
    &-U_{3}^{1}U_{2}^{1}U_{1}^{1}(\ket{\phi_{x}^{in}}\bra{\phi_{x}^{in}}\otimes\ket{000}_{hid}\bra{000}){U_{1}^{1}}^{\dagger}{U_{2}^{1}}^{\dagger}K_{2}^{1}{U_{3}^{1}}^{\dagger}
    \nonumber\\
    &-U_{3}^{1}U_{2}^{1}U_{1}^{1}
\left(\ket{\phi_{x}^{in}}\bra{\phi_{x}^{in}}\otimes\ket{000}_{hid}\bra{000}\right) {U_{1}^{1}}^{\dagger}K_{1}^{1}{U_{2}^{1}}^{\dagger}{U_{3}^{1}}^{\dagger})
+o(\epsilon^2)\nonumber\\
    =&\rho_x^{1_{out}}(s)+R(\epsilon)
    +o(\epsilon^2),
\end{align}
Here the second inequality is due to the Taylor's Formula of the exponential function, and the second term in the second inequality denoted as $R(\epsilon)$.
Next going through the residual block structure to $\rho_x^{1_{out}}(s+\epsilon)$, we obtain the real input state of layer with $l=2$:
\begin{align}\label{a4}
 \rho_x^{2_{in}}(s+\epsilon)=& \rho_x^{1_{out}}(s+\epsilon)+\ket{\phi_x^{in}}\bra{\phi_x^{in}}\otimes \ket{0}\bra{0}\nonumber\\
=&(\rho_x^{1_{out}}(s)+\ket{\phi_x^{in}}\bra{\phi_x^{in}}\otimes \ket{0}\bra{0})+R(\epsilon)
+o(\epsilon^2)\nonumber\\
=&\rho_x^{2_{in}}(s)+R(\epsilon)+o(\epsilon^2).
\end{align}
So the final output state of Res-HQCNN $[2,\tilde{3},2]$ using $U^2(s+\epsilon)$ in Eq.(\ref{a2}) is
\begin{align}\label{a5}
    \rho_x^{2_{out}}(s+\epsilon)=&\trace_{hid}(e^{i\epsilon K_{2}^{2}}U_{2}^{2} e^{i\epsilon K_{1}^{2}}U_{1}^{2}  (\rho_x^{2_{in}}(s+\epsilon)\otimes\ket{00}_{out}\bra{00})
    {U_{1}^{2}}^{\dagger}e^{-i\epsilon K_{1}^{2}} {U_{2}^{2}}^{\dagger}e^{-i\epsilon K_{2}^{2}} )\nonumber\\
    =&\trace_{hid}(e^{i\epsilon K_{2}^{2}}U_{2}^{2} e^{i\epsilon K_{1}^{2}}U_{1}^{2} (\rho_x^{2_{in}}(s)\otimes\ket{00}_{out}\bra{00})
    {U_{1}^{2}}^{\dagger}e^{-i\epsilon K_{1}^{2}} {U_{2}^{2}}^{\dagger}e^{-i\epsilon K_{2}^{2}} )\nonumber\\
    &+\trace_{hid}(e^{i\epsilon K_{2}^{2}}U_{2}^{2} e^{i\epsilon K_{1}^{2}}U_{1}^{2}  (R(\epsilon)\otimes
    \ket{00}_{out}\bra{00}){U_{1}^{2}}^{\dagger}e^{-i\epsilon K_{1}^{2}}{U_{2}^{2}}^{\dagger}e^{-i\epsilon K_{2}^{2}})
    +o(\epsilon^2)\nonumber\\
=&(\rho_x^{2_{out}}(s)+A_1)+A_2+o(\epsilon^2),
\end{align}
Here
\begin{align*}
A_1=&i\epsilon\trace_{hid}(U_2^{2}K_1^{2}U_1^{2}(\rho_x^{2_{in}}(s)\otimes\ket{00}_{out}\bra{00}){U_{1}^{2}}^{\dagger}{U_{2}^{2}}^{\dagger}+K_2^{2}U_2^{2}U_1^{2}(\rho_x^{2_{in}}(s)\otimes\ket{00}_{out}\bra{00}){U_{1}^{2}}^{\dagger}{U_{2}^{2}}^{\dagger}\nonumber\\
&-U_2^{2}U_1^{2}(\rho_x^{2_{in}}(s)\otimes\ket{00}_{out}\bra{00}){U_{1}^{2}}^{\dagger}{U_{2}^{2}}^{\dagger}K_2^{2}-U_2^{2}U_1^{2}(\rho_x^{2_{in}}(s)\otimes\ket{00}_{out}\bra{00}){U_{1}^{2}}^{\dagger}K_1^{2}{U_{2}^{2}}^{\dagger}),
\end{align*}
\begin{align*}
A_2=&i\epsilon\trace_{in,hid}(U_2^{2}U_1^{2}U_{3}^{1}U_{2}^{1}K_{1}^{1}U_{1}^{1}(\ket{\phi_x^{in}}\bra{\phi_x^{in}}\otimes\ket{00000}_{hid,out}\bra{00000}){U_{1}^{1}}^{\dagger}{U_{2}^{1}}^{\dagger}{U_{3}^{1}}^{\dagger}{U_{1}^{2}}^{\dagger}{U_{2}^{2}}^{\dagger}\nonumber\\
&+U_2^{2}U_1^{2}U_{3}^{1}K_{2}^{1}U_{2}^{1}U_{1}^{1}(\ket{\phi_x^{in}}\bra{\phi_x^{in}}\otimes\ket{00000}_{hid,out}\bra{00000}){U_{1}^{1}}^{\dagger}{U_{2}^{1}}^{\dagger}{U_{3}^{1}}^{\dagger}{U_{1}^{2}}^{\dagger}{U_{2}^{2}}^{\dagger}\nonumber\\
&+U_2^{2}U_1^{2}K_{3}^{1}U_{3}^{1}U_{2}^{1}U_{1}^{1}(\ket{\phi_x^{in}}\bra{\phi_x^{in}}\otimes\ket{00000}_{hid,out}\bra{00000}){U_{1}^{1}}^{\dagger}{U_{2}^{1}}^{\dagger}{U_{3}^{1}}^{\dagger}{U_{1}^{2}}^{\dagger}{U_{2}^{2}}^{\dagger}\nonumber\\
&-U_2^{2}U_1^{2}U_{3}^{1}U_{2}^{1}U_{1}^{1}(\ket{\phi_x^{in}}\bra{\phi_x^{in}}\otimes\ket{00000}_{hid,out}\bra{00000}){U_{1}^{1}}^{\dagger}{U_{2}^{1}}^{\dagger}{U_{3}^{1}}^{\dagger}K_{3}^{1}{U_{1}^{2}}^{\dagger}{U_{2}^{2}}^{\dagger})\nonumber\\
&-U_2^{2}U_1^{2}U_{3}^{1}U_{2}^{1}U_{1}^{1}(\ket{\phi_x^{in}}\bra{\phi_x^{in}}\otimes\ket{00000}_{hid,out}\bra{00000}){U_{1}^{1}}^{\dagger}{U_{2}^{1}}^{\dagger}K_{2}^{1}{U_{3}^{1}}^{\dagger}{U_{1}^{2}}^{\dagger}{U_{2}^{2}}^{\dagger})\nonumber\\
&-U_2^{2}U_1^{2}U_{3}^{1}U_{2}^{1}U_{1}^{1}(\ket{\phi_x^{in}}\bra{\phi_x^{in}}\otimes\ket{00000}_{hid,out}\bra{00000}){U_{1}^{1}}^{\dagger}K_{1}^{1}{U_{2}^{1}}^{\dagger}{U_{3}^{1}}^{\dagger}{U_{1}^{2}}^{\dagger}{U_{2}^{2}}^{\dagger}).
\end{align*}
Now based on the preparation above and choose $x=v,w$, we can calculate the mathematical derivative function of the cost function
\begin{align}\label{a6}
\frac{dC_{G}}{ds}=&\lim_{\epsilon \to 0}{\frac{C_{G}(s+\epsilon)-C_{G}(s)}{\epsilon}}\nonumber\\
=&\lim_{\epsilon \to 0}\frac{\sum_{v,w\in V}[A]_{v,w}D_{HS}\left(\mathcal{R}(\rho_v)(s+\epsilon),\mathcal{R}(\sigma_w)(s+\epsilon)\right)}{\epsilon}\nonumber\\
&-\lim_{\epsilon \to 0}\frac{\sum_{v,w\in V}[A]_{v,w}D_{HS}\left(\mathcal{R}(\rho_v)(s),\mathcal{R}(\sigma_w)(s)\right)}{\epsilon}\nonumber\\
=&\lim_{\epsilon \to 0}\frac{\sum_{v,w\in V}[A]_{v,w}[\trace(\rho_v^{out}(s+\epsilon)-\rho_w^{out}(s+\epsilon))^2]}{\epsilon}\nonumber\\
&-\lim_{\epsilon \to 0}\frac{\sum_{v,w\in V}[A]_{v,w}[\trace(\rho_v^{out}(s)-\rho_w^{out}(s))^2]}{\epsilon}\nonumber\\
=&2\sum_{v,w\in V}[A]_{v,w}\trace(Id(3)\otimes(\rho_{v}^{out}-\rho_w^{out}))([i K_2^2,U_2^2U_1^2((\rho_{v}^{2_{in}}-\rho_w^{2_{in}})\otimes\ket{00}_{out}\bra{00}\nonumber\\
&{U_1^2}^{\dagger}{U_2^2}^{\dagger}]+U_2^2[i K_1^2,U_1^2((\rho_{v}^{2_{in}}-\rho_w^{2_{in}})\otimes\ket{00}_{out}\bra{00}{U_1^2}^{\dagger}]
{U_2^2}^{\dagger}+\cdots+U_2^1U_3^1U_1^2U_2^2\nonumber\\
&[i K_1^1, U_1^1((\rho_{v}^{1_{in}}-\rho_w^{1_{in}})\otimes\ket{000}_{hid}\bra{000})
{U_1^1}^{\dagger}]{U_2^2}^{\dagger}{U_1^2}^{\dagger}{U_3^1}^{\dagger}{U_2^1}^{\dagger}
)\nonumber\\
=&2\sum_{v,w\in V}[A]_{v,w}
\trace(([U_2^2U_1^2((\rho_{v}^{2_{in}}-\rho_w^{2_{in}})\otimes\ket{00}_{out}\bra{00}{U_1^2}^{\dagger}{U_2^2}^{\dagger},Id(3)\otimes(\rho_{v}^{out}-\rho_w^{out})]K_2^2\nonumber\\
&
+[U_1^2((\rho_{v}^{2_{in}}-\rho_w^{2_{in}})\otimes\ket{00}_{out}\bra{00}{U_1^2}^{\dagger},{U_2^2}^{\dagger}(Id(3)\otimes(\rho_{v}^{out}-\rho_w^{out}))U_2^2]K_1^2+\cdots+
[U_1^1\nonumber\\
&((\rho_{v}^{1_{in}}-\rho_w^{1_{in}})\otimes\ket{000}_{hid}\bra{000}){U_1^1}^{\dagger},{U_2^1}^{\dagger}{U_3^1}^{\dagger}{U_1^2}^{\dagger}{U_2^2}^{\dagger}(Id(3)\otimes(\rho_{v}^{out}-\rho_w^{out}))U_2^2U_1^2U_3^1U_2^1]K_1^1)\nonumber\\
=&2\sum_{v,w\in V}[A]_{v,w}\trace((M_2^2+N_2^2)K_2^2+
(M_1^2+N_1^2)K_1^2)+\frac{i}{N}\sum_{x=1}^{N}\trace(M_3^1K_3^1+
M_2^1K_2^1+M_1^1K_1^1),
\end{align}
where 
\begin{align*}
M_1^1=&[U_1^1\left((\rho_v^{in}-\rho_w^{in})\otimes\ket{000}_{hid}\bra{000}\right){U_1^1}^{\dagger},{U_2^1}^{\dagger}{U_3^1}^{\dagger}{U_1^2}^{\dagger}{U_2^2}^{\dagger}\left(Id(3)\otimes(\rho_{v}^{out}-\rho_w^{out})\right)U_2^2U_1^2U_2^1U_3^1],
\end{align*}
\begin{align*}
M_2^1=&[U_2^1U_1^1\left((\rho_v^{in}-\rho_w^{in})\otimes\ket{000}_{hid}\bra{000}\right){U_1^1}^{\dagger}{U_2^1}^{\dagger},{U_3^1}^{\dagger}{U_1^2}^{\dagger}{U_2^2}^{\dagger}\left(Id(3)\otimes(\rho_{v}^{out}-\rho_w^{out})\right)U_2^2U_1^2U_3^1],
\end{align*}
\begin{align*}
M_3^1=&[U_3^1U_2^1U_1^1\left((\rho_v^{in}-\rho_w^{in})\otimes\ket{000}_{hid}\bra{000}\right){U_1^1}^{\dagger}{U_2^1}^{\dagger}{U_3^1}^{\dagger},{U_1^2}^{\dagger}{U_2^2}^{\dagger}\left(Id(3)\otimes(\rho_{v}^{out}-\rho_w^{out})\right)U_2^2U_1^2],
\end{align*}
\begin{align*}
M_1^2=&[U_1^2U_3^1U_2^1U_1^1\left((\rho_v^{in}-\rho_w^{in})\otimes\ket{000}_{hid}\bra{000}\right){U_1^1}^{\dagger}{U_2^1}^{\dagger}{U_3^1}^{\dagger}{U_1^2}^{\dagger},{U_2^2}^{\dagger}\left(Id(3)\otimes(\rho_{v}^{out}-\rho_w^{out})\right)U_2^2],
\end{align*}
\begin{align*}
M_2^2=&[U_2^2U_1^2U_3^1U_2^1U_1^1\left((\rho_v^{in}-\rho_w^{in})\otimes\ket{000}_{hid}\bra{000}\right){U_1^1}^{\dagger}{U_2^1}^{\dagger}{U_3^1}^{\dagger}{U_1^2}^{\dagger}{U_2^2}^{\dagger},\left(Id(3)\otimes(\rho_{v}^{out}-\rho_w^{out})\right)],
\end{align*}
\begin{align*}
N_1^2=&[U_1^2\left((\rho_v^{in}-\rho_w^{in})\otimes\ket{000}_{hid}\bra{000}\right){U_1^2}^{\dagger},{U_2^2}^{\dagger}\left(Id(3)\otimes(\rho_{v}^{out}-\rho_w^{out})\right)U_2^2],
\end{align*}
\begin{align*}
N_2^2=&[U_2^2U_1^2((\rho_v^{in}-\rho_w^{in})\otimes\ket{000}_{hid}\bra{000}){U_1^2}^{\dagger}{U_2^2}^{\dagger},Id(3)\otimes(Id(3)\otimes(\rho_{v}^{out}-\rho_w^{out}))],
\end{align*}
Here the third equality in Eq.(\ref{a6}) is due to Eq.(\ref{a5}) and the definition of $D_{HS}(\rho,\sigma)=\trace[(\rho-\sigma)^2]$. The coefficient $2$ comes from $tr(AB)=tr(BA)$ and $\rho^2=\rho \rho^{\dagger}$. The fourth equality also comes from $tr(AB)=tr(BA)$. Using Eq.(\ref{a4}), we get the last equality.
$[A,B]=A\times B-A\times B$ is a commutator operator.

Next we want to find the maximum of the cost function $\frac{dC_G}{ds}$ using a  Lagrange multiplier $\lambda$, which is a real number. Firstly, when $l=1$, since unitary $U_j^1(s)$ works on three qubits, then  $K_j^1$ is
\begin{align}\label{a7}
K_j^1(s)=\sum_{\alpha_1,\alpha_2,\beta}{K_j^1}_{\alpha_1,\alpha_2,\beta}(s)(\sigma^{\alpha_1}\otimes\sigma^{\alpha_2}\otimes\sigma^{\beta}),
\end{align}
Here $\sigma$ is Pauli matrix in single qubit. $\alpha_1$ and $\alpha_2$ represent the qubits in the input layer and $\beta$ represents the current qubit of unitary $U_j^1$ in the hidden layer. Then when $j=1$, the analysis above leads us to solve a maximization problem:
\begin{align}
\begin{split}
&\max_{K_1^1,\alpha_1,\alpha_2,\beta}(\frac{dC_G}{ds}-\lambda\sum_{\alpha_1,\alpha_2,\beta}{{K_1^1}_{\alpha_1,\alpha_2,\beta}}^2)\nonumber\\
&=\max_{K_1^1,\alpha_1,\alpha_2,\beta}(2\sum_{v,w\in V}[A]_{v,w}\trace((M_2^2+N_2^2)K_2^2+(M_1^2+N_1^2)K_1^2)\nonumber\\
&+
2\sum_{v,w\in V}[A]_{v,w}\trace(M_3^1K_3^1+M_2^1K_2^1+M_1^1K_1^1)
- \lambda\sum_{\alpha_1,\alpha_2,\beta}{{K_1^1}_{\alpha_1,\alpha_2,\beta}}^2 )\nonumber\\
&=\max_{K_1^1,\alpha_1,\alpha_2,\beta}(2\sum_{v,w\in V}[A]_{v,w}\trace_{\alpha_1,\alpha_2,\beta}[\trace_{rest}((M_1^2+N_1^2)K_1^2+(M_2^2+N_2^2)K_2^2
+M_3^1K_3^1\nonumber\\
&+M_2^1K_2^1)+\trace_{rest}(M_1^1)\sum_{\alpha_1,\alpha_2,\beta}{K_1^1}_{\alpha_1,\alpha_2,\beta}(\sigma^{\alpha_1}
\otimes\sigma^{\alpha_2}\otimes\sigma^{\beta})]
- \lambda\sum_{\alpha_1,\alpha_2,\beta}{{K_1^1}_{\alpha_1,\alpha_2,\beta}}^2 ).\nonumber\\
\end{split}
\end{align}
Taking the derivative of ${K_j^1}_{\alpha_1,\alpha_2,\beta}$ to be zero yields
\begin{align}
{K_1^1}_{\alpha_1,\alpha_2,\beta}=&\frac{i}{\lambda }\sum_{x=1}^{N}\trace_{\alpha_1,\alpha_2,\beta}((\trace_{rest}(M_1^1))(\sigma^{\alpha_1}\otimes
\sigma^{\alpha_2}\otimes
\sigma^{\beta})).
\end{align}
This above equation further leads to the matrix:
\begin{align}
K_1^1=&\sum_{\alpha_1,\alpha_2,\beta}{K_1^1}_{\alpha_1,\alpha_2,\beta}(\sigma^{\alpha_1}\otimes\sigma^{\alpha_2}\otimes\sigma^{\beta})\nonumber\\
=&2\sum_{v,w\in V}[A]_{v,w}\sum_{\alpha_1,\alpha_2,\beta}\trace_{\alpha_1,\alpha_2,\beta}((\trace_{rest}(M_1^1))(\sigma^{\alpha_1}\otimes\sigma^{\alpha_2}\otimes\sigma^{\beta})
\times(\sigma^{\alpha_1}\otimes\sigma^{\alpha_2}\otimes\sigma^{\beta}))\nonumber\\
=&\frac{8i}{\lambda}\sum_{x=1}^{N}\trace_{rest}(M_1^1).
\end{align}
Here $\frac{1}{\lambda}$ is the learning rate in this paper.
Analogously, we find out the formulas for $K_2^1$ and $K_3^1$:
$K_2^1=\frac{8i}{\lambda}\sum_{x=1}^{N}\trace_{rest}(M_2^1);$
$K_3^1=\frac{8i}{\lambda}\sum_{x=1}^{N}\trace_{rest}(M_3^1).$

When $l=2$ in Res-HQCNN $[2,\tilde{3},2]$, we can get $K_q^2$
through going on the similar method for $K_j^1$:
$$K_q^2=\frac{16i}{\lambda }\sum_{x=1}^{N}\trace_{rest}(M_q^2+N_q^2),\quad q=1,2.$$

\end{proof}

\end{appendices}

\end{document}